\documentclass[twocolumn,aps,pra,notitlepage,superscriptaddress,longbibliography]{revtex4-1}
\usepackage[colorlinks, linkcolor = ForestGreen, anchorcolor = blue, citecolor = magenta]{hyperref}
\usepackage{amsmath,amssymb,amsthm,mathtools,mathrsfs}
\usepackage[table,dvipsnames]{xcolor}
\usepackage{tikz}
\usepackage{adjustbox}
\usepackage{dsfont}
\usepackage[normalem]{ulem}
\usepackage{enumerate}
\usepackage[all]{xy} 

\theoremstyle{plain}
\newtheorem{theorem}{Theorem}
\newtheorem{lemma}{Lemma}
\newtheorem{proposition}{Proposition}

\theoremstyle{definition}
\newtheorem{definition}{Definition}

\newtheorem{example}{Example}

\newcommand{\bd}{\begin{definition}}
\newcommand{\ed}{\end{definition}}
\newcommand{\bt}{\begin{theorem}}
\newcommand{\et}{\end{theorem}}
\newcommand{\bn}{\begin{proposition}}
\newcommand{\en}{\end{proposition}}
\newcommand{\be}{\begin{equation}}
\newcommand{\ee}{\end{equation}}
\newcommand{\blem}{\begin{lemma}}
\newcommand{\elem}{\end{lemma}}
\newcommand{\bx}{\begin{example}}
\newcommand{\ex}{\end{example}}
\newcommand{\bprf}{\begin{proof}}
\newcommand{\eprf}{\end{proof}}

\newcommand\define[1]{\emph{\textbf{#1}}}

\DeclareMathAlphabet{\mathpzc}{OT1}{pzc}{m}{it} 
 \DeclareFontFamily{OT1}{pzc}{}
 \DeclareFontShape{OT1}{pzc}{m}{it}{ <-> s*[1.2] pzcmi7t }{}
 \DeclareMathAlphabet{\mathpzc}{OT1}{pzc}{m}{it}
 \newcommand{\Alg}[1]{\mathpzc{#1}}

\newcommand{\map}[1]{\mathcal{#1}}

\newcommand{\Ad}{\mathrm{Ad}}
\def\R{{{\mathbb R}}}
\def\C{{{\mathbb C}}}
\newcommand{\<}{\langle}
\renewcommand{\>}{\rangle}

\newcommand{\Tr}{\operatorname{Tr}}

\newcommand{\Choi}{\mathscr{C}}

\def\A{\Alg{A}}
\def\B{\Alg{B}}
\def\M{\mathbb{M}}

\def\E{\mathcal{E}}
\def\H{\mathcal{H}}

\def\O{\mathscr{O}}

\def\J{\mathscr{J}}
\def\C{\mathbb{C}}

\def\R{\mathbb{R}}

\begin{document}									
\preprint{APS/123-QED}

\title{Partial transpose as a space-time swap}

\author{James Fullwood}
\email{fullwood@hainanu.edu.cn}
\affiliation{School of Mathematics and Statistics, Hainan University, Haikou, Hainan Province, 570228, China.}
\author{Junxian Li}
\affiliation{School of Mathematics and Statistics, Hainan University, Haikou, Hainan Province, 570228, China.}

\date{\today}

\begin{abstract}
While the partial transpose operation appears in many fundamental results in quantum theory---such as the Peres-Horodecki criterion for entanglement detection---a physical interpretation of the partial transpose is lacking. In this work, we show that a partial transpose of a bipartite density operator is a two-time \emph{pseudo}-density operator, which by definition encodes temporal correlations associated with two-point sequential measurement scenarios. As such, it follows that partial transposition admits a precise physical interpretation as mapping spatial correlations to temporal correlations, thus swapping the roles of space and time for bipartite quantum systems. For maximally entangled qubits, we show that partial transposition maps spatial correlations which violate Bell inequalities to causal correlations which cannot be replicated by spacelike separated systems, thus further solidifying the interpretation of partial transpose as a space-time swap. As it is known that gravitational effects swap the roles of space and time inside a black hole, our results suggest that at a quantum mechanical level, a traversal of a black hole's event horizon by a bipartite quantum system may be described by a partial transpose.
\end{abstract}

	\maketitle

\section{Introduction}

Quantum information theory uses completely positive, trace-preserving (CPTP) linear maps to model the dynamics of open quantum systems. Typically the system of interest $A$ is interacting with a large environment such as a heat bath or a sea of photons, and in such a case, the CPTP map then models the noise induced by the system-environment interaction. The complete positivity of a CPTP map $\E$ then ensures that if we view the system $A$ as being a subsystem of a composite system $AC$, then the mapping $\E\otimes \mathcal{I}_C$ also represents a valid physical transformation of the composite system $AC$, where $\mathcal{I}_C$ is the identity channel on $C$. 

While the map $\rho\mapsto \rho^T$ taking a state to its transpose is a positive map, it is not completely positive, since tensoring the transpose map with an identity map---and thus forming a \emph{partial} transpose---is no longer positive. As such, it follows that even though the transpose map takes states to states, such a mapping does not represent an implementable quantum operation on a physical system. Although this fact is well known, a more conceptual explanation for why a partial transpose map is not positive is lacking. Moreover, as the partial transpose plays a fundamental role in the study of entanglement~\cite{Peres_96,Horodecki_96,Werner_2001,Hillery_2006,Toth_2009}, a precise physical interpretation of the partial transpose may lead to new insights on long-standing open problems, such as whether or not there exists bound entangled states with negative partial transpose~\cite{Horodecki_2022}.  

In Ref.~\cite{Marletto_2020}, it was shown that a partial transpose of a Bell state corresponding to two maximally entangled qubits is a \emph{pseudo-density matrix}, which is a generalization of a density matrix which encodes temporal correlations between timelike separated systems~\cite{FJV15,HLiu_2025,MVVAPGDG21,Marletto_2019,Pisar19,ZPTGVF18,song23,FuPa24,Liu_2024,Liu_2025}. It then follows that in the context of two maximally entangled qubits, a partial transpose has the effect of transforming spatial correlations into timelike correlations, thus swapping the roles of space and time. 

In this work, we show that viewing a partial transpose as a space-time swap holds not only for two maximally entangled qubits, but also for bipartite systems consisting of an arbitrary number of qubits, regardless of whether or not the two systems are entangled. In particular, we show that if $\rho_{AB}$ is a bipartite density matrix representing the state of a composite system $AB$, then the partial transpose $\rho_{AB}^{T_B}\equiv T_B(\rho_{AB})$ encodes temporal correlations associated with a two-point sequential measurement scenario consisting of the following protocol: (i) Alice prepares the state $\rho_A=\Tr_{B}[\rho_{AB}]$. (ii) Alice measures a multi-qubit Pauli observable $\sigma_{\alpha}$. (iii) Alice sends the updated state to Bob via a quantum channel $\E$ determined by the bipartite state $\rho_{AB}$. (iv) Bob measures a multi-qubit Pauli observable $\sigma_{\beta}$ on the output of the channel $\E$. If we then form the two-time correlator $\<\sigma_{\alpha},\sigma_{\beta}\>$ corresponding to the expectation value of the product of Alice and Bob's measurements, we show  
\be \label{CRLXFCT87}
\<\sigma_{\alpha},\sigma_{\alpha}\>=\Tr\Big[\rho_{AB}^{T_B}\,(\sigma_{\alpha}\otimes \sigma_{\beta})\Big]\, ,
\ee
thus the partial transpose $\rho_{AB}^{T_B}$ may be viewed as a spatiotemporal quantum state encoding the temporal correlations between the timelike separated systems $A$ and $B$. 
As for the partial transpose $\rho_{AB}^{T_A}\equiv T_A(\rho_{AB})$, we show that it also encodes temporal correlations between $A$ and $B$, but in the slightly modified scenario where Alice instead prepares the state $\rho_A^T$, and the channel responsible for the dynamics of the system between measurements is $\overline{\E}=T\circ \E\circ T$, where $T$ denotes the transpose map. 

If we view the channel $\overline{\E}$ as the channel $\E$ but with an opposite orientation in time, then a compelling picture emerges where the partial transpositions $T_B$ and $T_A$ correspond to reflections of the composite system $AB$ about the future and past light cones of an auxiliary system, a sketch of which is provided in Fig.~\ref{F1}. On the other hand, if we fix the systems $A$ and $B$ and view the spacetime swap as a passive transformation, then the light cone in Fig.~\ref{F1} is rotated $90^{\text{o}}$ under the partial transpositions $T_A$ and $T_B$, which is reminiscent of the way in which gravitational effects inside a black hole tilt light cones toward the singularity. As the Schwarzschild metric in particular literally swaps the roles of space and time, our results suggest that at a quantum mechanical level, a traversal of a Schwarzschild event horizon by a bipartite quantum system may be described by a partial transpose.

\begin{figure}\label{F1}
    \centering
    \includegraphics[width=0.7\columnwidth]{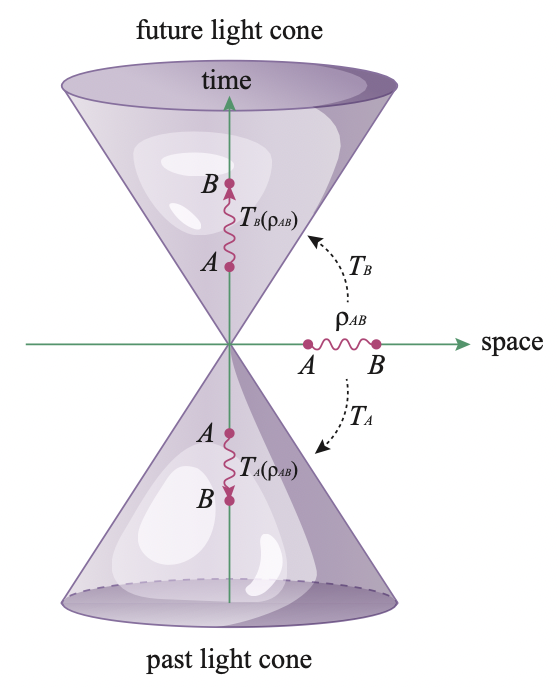}
    \caption{The partial transpositions $T_A$ and $T_B$ as reflections about the future and past light cones of an auxiliary system.}
    \label{fig1}
\end{figure}

As an illustration of our results, we show that in the context of Bell scenarios where Alice and Bob are performing randomized measurements on a maximally entangled state of two qubits, the partial transpositions $T_B$ and $T_A$ both map spatial correlations which violate Bell inequalities to temporal correlations which violate temporal analogs of Bell inequalities~\cite{BTCV04,Fritz10}. In particular, we show that partial transposition maps non-local correlations in space to causal correlations which may not be replicated by spacelike separated systems, thus further solidifying the interpretation of the partial transpose as a space-time swap.

\section{Two-time correlators}

Throughout this work, $A$ and $B$ will denote multi-qubit systems with Hilbert spaces $\H_A$ and $\H_B$, respectively. The algebra of linear operators on $\H_A$ will be denoted by $\A$ and the algebra of linear operators on $\H_B$ will be denoted by $\B$. A linear map $\E:\A\to \B$ which is completely positive and trace-preserving (CPTP) will be referred to as a \define{quantum channel}.

If system $A$ consists of $m$ qubits, then an observable of the form $\sigma_{\alpha}=\sigma_{\alpha_1}\otimes \cdots \otimes \sigma_{\alpha_{m}}$ with $\alpha_i\in \{0,\ldots,3\}$ will be referred to as a \define{Pauli observable} on $A$, and similarly for system $B$. A \define{two-point sequential measurement} of Pauli observables on $A$ and $B$ then consists of the following protocol:
\begin{itemize}
\item
Alice prepares system $A$ in state $\rho_A$.
\item
Alice then measures a Pauli observable $\sigma_{\alpha}$.
\item
Alice sends the output of her measurement to Bob via a quantum channel $\E:\A\to \B$.
\item
Bob measures a Pauli observable $\sigma_{\beta}$ on the output of the channel $\E:\A\to \B$.
\end{itemize}
Such a two-point sequential measurement (TPSM) scenario will be denoted by $(\rho_A,\sigma_{\alpha},\E,\sigma_{\beta})$, and in such a case, the \define{two-time correlator} corresponding to the expectation value of the \emph{product} of Alice and Bob's measurements is the real number $\<\sigma_{\alpha},\sigma_{\beta}\>$ given by
\be \label{2TXCRX57}
\<\sigma_{\alpha},\sigma_{\beta}\>=\Tr[\E(\Pi_{\alpha}^{+}\rho_A\Pi_{\alpha}^{+})\sigma_{\beta}]-\Tr[\E(\Pi_{\alpha}^{-}\rho_A\Pi_{\alpha}^{-})\sigma_{\beta}]\, ,
\ee
where $\Pi_{\alpha}^{\pm}$ are the projectors onto the $\pm 1$-eigenspaces of $\sigma_{\alpha}$, so that $\sigma_{\alpha}=\Pi_{\alpha}^{+}-\Pi_{\alpha}^{-}$. The two-time correlator $\<\sigma_{\alpha},\sigma_{\beta}\>$ will be denoted by $\<\sigma_{\alpha},\sigma_{\beta}\>_{(\rho_A,\E)}$ when it is important to emphasize its dependence on the initial state $\rho_A$ and the channel $\E:\A\to \B$. 

Due to the fact that the collection of Pauli observables $\{\sigma_{\alpha}\otimes \sigma_{\beta}\}$ forms an orthogonal basis of observables on the composite system $AB$, it follows that there exists a unique operator $\varrho_{AB}$ such that
\be \label{PDMEXT71}
\<\sigma_{\alpha},\sigma_{\beta}\>=\Tr[\varrho_{AB}(\sigma_{\alpha}\otimes \sigma_{\beta})]
\ee
for all Pauli observables $\sigma_{\alpha}$ and $\sigma_{\beta}$. It is straightforward to show that the operator $\varrho_{AB}$ is then given by
\be \label{PDMDFX87}
\varrho_{AB}=\E\star \rho_A\equiv \frac{1}{2}\Big\{\rho_A\otimes \mathds{1},\J[\E]\Big\}\, ,
\ee
where $\{\cdot,\cdot\}$ denotes the anticommutator and $\J[\E]=\sum_{i,j}|i\>\<j|\otimes \E(|j\>\<i|)$ is the \define{Jamio\l kowski operator} of the channel $\E:\A\to \B$. While the operator $\varrho_{AB}$ as given by \eqref{PDMDFX87} is Hermitian and of unit trace, it is not positive semi-definite in general, and thus will be referred to as a \define{pseudo-density operator}. 

If the pseudo-density operator $\varrho_{AB}$ is in fact positive semi-definite, then the correlations between $A$ and $B$ admit a dual description where $A$ and $B$ are spacelike separated and Alice and Bob are performing measurements in parallel. Such states will be referred to as \define{dual states}. Otherwise, the negative eigenvalues of a pseudo-density operator serve as a witness to temporal correlations which may not be realized by spacelike separated systems~\cite{FJV15}. 

\section{The partial transpose theorem}

In this section we present our main result, namely, that the partial transpositions of a bipartite density operator $\rho_{AB}$ are both pseudo-density operators, and thus encode temporal correlations associated with a TPSM scenario. For this, we fix orthonormal bases $\{|i\>\}$ of $\H_A$ and $\{|j\>\}$ of $\H_B$, so that we may identify $\A$ and $\B$ with algebras of matrices, thus yielding a well-defined notion of the partial transpositions $\rho_{AB}^{T_A}$ and $\rho_{AB}^{T_B}$ for every operator $\rho_{AB}\in \A\otimes \B$. 

\bt[The Partial Transpose Theorem] \label{MTX1}
Let $\rho_{AB}\in \A\otimes \B$ be a bipartite density operator. Then there exists a quantum channel $\E:\A\to \B$ such that 
\be \label{PTTXE57}
\rho_{AB}^{T_B}=\E\star \rho_A \quad \text{and} \quad \rho_{AB}^{T_A}=\overline{\E}\star \rho_A^T\, ,
\ee
where $\rho_A=\Tr_B[\rho_{AB}]$ and $\overline{\E}$ is the complex conjugate of $\E$, which is the channel given by $\overline{\E}=T\circ \E\circ T$, with $T$ the transpose map. Moreover, if $\rho_A$ is of full rank, then the channel $\E$ is unique.
\et

\bprf[Sketch of the proof]
To show there exists a quantum channel $\E:\A\to \B$ such that $\rho_{AB}^{T_B}=\E\star \rho$, we first show that there exists a solution $X$ of the Sylvester equation 
\[
(\rho_A\otimes \mathds{1})X+X(\rho_A\otimes \mathds{1})=2\rho_{AB}^{T_B}\, ,
\]
from which it follows that $\rho_{AB}^{T_B}=\E\star \rho$, where $\E=\J^{-1}[X]$ and $\J^{-1}$ is the inverse of the Jamio\l kowski isomorphism $\J$. To show $\E$ is CPTP, we use the fact that $\E$ is completely positive (CP) if and only if the Choi matrix $\mathscr{C}[\E]=X^{T_A}$ is positive semi-definite~\cite{Ch75}, together with the fact that $\E$ is trace-preserving if and only if $\Tr_B[X]=\mathds{1}$. The full details showing that $\mathscr{C}[\E]$ is positive semi-definite and $\Tr_B[X]=\mathds{1}$ are given in the Supplemental Material.  
\eprf

When combined with Eqs.~\eqref{PDMEXT71} and \eqref{PDMDFX87}, the Partial Transpose Theorem yields a precise physical interpretation of the partial transpose maps $T_A$ and $T_B$. In particular, it follows from the Partial Transpose Theorem that for any bipartite density operator $\rho_{AB}$, there exists a quantum channel $\E:\A\to \B$ such that 
\be \label{EQX1}
\<\sigma_{\alpha},\sigma_{\beta}\>_{(\rho_A,\E)}=\Tr\Big[\rho_{AB}^{T_B}(\sigma_{\alpha}\otimes \sigma_{\beta})\Big]
\ee
for all TPSM scenarios $(\rho_A,\sigma_{\alpha},\E,\sigma_{\beta})$, and
\be \label{EQX2}
\<\sigma_{\alpha},\sigma_{\beta}\>_{(\rho_A^T,\overline{\E})}=\Tr\Big[\rho_{AB}^{T_A}(\sigma_{\alpha}\otimes \sigma_{\beta})\Big]
\ee
for all TPSM scenarios $(\rho_A^T,\sigma_{\alpha},\overline{\E},\sigma_{\beta})$. As such, the partial transpositions $T_A$ and $T_B$ both transform the spatial correlations encoded by a bipartite density operator $\rho_{AB}$ into temporal correlations associated with explicit TPSM scenarios determined by $\rho_{AB}$.

We note that it immediately follows from the Partial Transpose Theorem that a bipartite state $\rho_{AB}$ with positive partial transpose is necessarily a dual state, i.e., $\rho_{AB}=\E\star \rho_A$ for some quantum channel $\E$. Indeed, if $\rho_{AB}^{T_B}\geq 0$, then there exists a bipartite density operator $\sigma_{AB}$ such that $\rho_{AB}^{T_B}=\sigma_{AB}$. Since by the Partial Transpose Theorem we have $\sigma_{AB}^{T_B}=\E\star \rho_A$ for some quantum channel $\E$, we then have
\[
\rho_{AB}=\sigma_{AB}^{T_B}=\E\star \rho_A\, ,
\]
thus $\rho_{AB}$ is a dual state. This fact was also recently discovered in Ref.~\cite{song25}, which gives a nice characterization of dual states by giving an explicit decomposition of the channel that realizes their temporal description.

\section{Swapping spatial and temporal Bell inequality violations}
In this section, we show that partial transposition maps correlations that violate spatial Bell inequalities to correlations that violate temporal Bell inequalities, thus mapping non-local correlations in space to causal correlations in time. For this, we first consider a spatial Bell scenario where Alice and Bob perform spacelike separated measurements on a physical system with binary outcomes labeled by $\pm 1$, with Alice and Bob each performing one of two measurements according to some randomization procedure such as flipping a coin. By assuming that nature satisfies local realism~\cite{EPR,Bell64,Fullwood_2025}---namely, that there are no superluminal influences between Alice and Bob's measurements, and that the values of their measurement outcomes have well-defined, pre-existing values independent of measurement--- one may derive constraints on the joint statistics of Alice and Bob's measurements often referred to as \emph{Bell inequalities}. In particular, if we denote Alice's measurements by $Q$ and $R$ and Bob's measurements by $S$ and $T$, then the assumption of local realism implies the Bell inequality
\be \label{BLXEQX71}
|E_{QS}+E_{RS}+E_{RT}-E_{QT}|\leq 2\, ,
\ee
where $E_{QS}$ is the expectation value of the product of Alice and Bob's measurements for runs of the experiment where Alice measures $Q$ and Bob measures $S$, and similarly for $E_{RS}$, $E_{RT}$ and $E_{QT}$.

It is well known that if Alice and Bob share a maximally-entangled state of two qubits, then it is possible to choose $Q$, $R$, $S$ and $T$ to be Pauli observables such that the inequality \eqref{BLXEQX71} is violated, thus implying that nature is not locally real. For example, such a maximally entangled state $\rho_{AB}$ may be written as
\be \label{MAXENT91}
\rho_{AB}=\frac{1}{4}\left(\mathds{1}\otimes \mathds{1}+\sum_{i=1}^{3}a_i(\sigma_i\otimes \sigma_i)\right)\, ,
\ee
where $\sigma_1=\sigma_x$, $\sigma_2=\sigma_y$, and $\sigma_3=\sigma_z$, with $a_1=a_3=\pm 1$ and $a_2=-1$. When $a_1=a_3=1$ the spins $\sigma_i$ on Alice and Bob's systems are perfectly correlated for $i=1,3$, and when $a_1=a_3=-1$ the spins $\sigma_i$ are perfectly \emph{anti}-correlated for $i=1,3$. The spins $\sigma_2$ are always perfectly anti-correlated since $a_2=-1$. After setting $Q=\sigma_3$, $R=\sigma_1$, $S=-(\sigma_3+\sigma_1)/\sqrt{2}$ and $T=(\sigma_3-\sigma_1)/\sqrt{2}$, we then have
\be \label{BEX571}
E_{QS}+E_{RS}+E_{RT}-E_{QT}=\Tr[\rho_{AB}\O]\, ,
\ee
where 
\begin{align*}
\O&=Q\otimes S+R\otimes S+R\otimes T-Q\otimes T \\
&=\frac{-2}{\sqrt{2}}(\sigma_1\otimes \sigma_1+\sigma_3\otimes \sigma_3)\, .
\end{align*}
Since $\O$ does not involve the term $\sigma_2\otimes \sigma_2$, it immediately follows from Eq.~\eqref{MAXENT91} that
\[
\Tr[\rho_{AB}\O]=-\sqrt{2}(a_1+a_3)=\pm 2\sqrt{2}\, ,
\]
where the final equation follows from the fact that $a_1=a_3=\pm 1$. It then follows from \eqref{BEX571} that
\[
E_{QS}+E_{RS}+E_{RT}-E_{QT}=\pm 2\sqrt{2}\, ,
\]
thus violating the Bell inequality \eqref{BLXEQX71}.

We now consider the temporal analog of such a Bell scenario, i.e., where Alice and Bob perform measurements according to a TPSM scenario $(\rho_A,\sigma_{\alpha},\E,\sigma_{\beta})$, where $\rho_{A}=\Tr_B[\rho_{AB}]$, and $\sigma_{\alpha}\in \{Q,R\}$ and $\sigma_{\beta}\in \{S,T\}$ are chosen according to the same randomization procedure as in the spatial case. As for the channel $\E$, we find by direct calculation using Theorem~\ref{MTX1} that it is the identity channel when $a_1=a_3=1$, and when $a_1=a_3=-1$, $\E$ is the channel which fixes $\mathds{1}$ and $\sigma_2$ while sending $\sigma_i$ to $-\sigma_i$ for $i=1,3$. In such a case we then have $E_{QS}=\<Q,S\>$, $E_{RS}=\<R,S\>$, $E_{RT}=\<R,T\>$ and $E_{QT}=\<Q,T\>$, where $\<\cdot,\cdot\>$ is the two-time correlator given by \eqref{2TXCRX57}. By Eq.~\eqref{EQX1}, we then have
\be \label{BEX57}
E_{QS}+E_{RS}+E_{RT}-E_{QT}=\Tr[\rho_{AB}^{T_B}\O]\, ,
\ee
where again $\O=-2(\sigma_1\otimes \sigma_1+\sigma_3\otimes \sigma_3)/\sqrt{2}$. Now since $\sigma_i^T=(-1)^{i+1}\sigma_i$, it follows that $\rho_{AB}^{T_B}=\rho_{AB}^{T_A}$ is obtained from $\rho_{AB}$ by simply setting $a_2=1$, which is reflected in the fact that the spins $\sigma_2$ on Alice and Bob's systems are now perfectly correlated in such a TPSM scenario. Moreover, since $\rho_{AB}^{T_B}$ is obtained from $\rho_{AB}$ by changing the value of $a_2$, and since $\Tr[\rho_{AB}\O]=-\sqrt{2}(a_1+a_3)$, it follows that 
\[
\Tr[\rho_{AB}^{T_B}\O]=\Tr[\rho_{AB}\O]=\pm 2\sqrt{2}\, ,
\]
which together with Eq.~\eqref{BEX57} shows that the Bell inequality \eqref{BLXEQX71} is again violated. 

We note however that even though the same Bell inequality is violated by the correlations encoded by $\rho_{AB}$ and $\rho_{AB}^{T_B}$, the correlations encoded by $\rho_{AB}$ and $\rho_{AB}^{T_B}$ are \emph{not} in fact equivalent. In particular, while in the spatial case the spins $\sigma_2$ on Alice and Bob's systems are perfectly anti-correlated, in the temporal case the spins $\sigma_2$ on Alice and Bob's systems are perfectly correlated. Moreover, since $\rho_{AB}^{T_B}$ has a negative eigenvalue, it follows that the correlations encoded by $\rho_{AB}^{T_B}$ may not be realized by spacelike separated qubits. In particular, while the correlations encoded by $\rho_{AB}$ are monogamous, the correlations encoded by $\rho_{AB}^{T_B}$ are polygamous, as Bob's qubit may not only be maximally correlated with Alice's qubit in the past, but also potentially with another qubit in its future~\cite{Wu_2025}. 

It is also important to point out that although in a temporal Bell scenario the notion of realism applies as in the spatial case, the notion of locality must be replaced with a notion of locality \emph{in time}. In the case of a TPSM scenario $(\rho_A,\sigma_{\alpha},\E,\sigma_{\beta})$, we propose that locality in time would mean that the state disturbance due to Alice's measurement of $\sigma_{\alpha}$ has no causal influence on the outcome Bob's measurement, which will depend not only on Alice's measurement outcome but also the channel $\E$ governing the dynamics between measurements. In particular, if $\E(\rho_A^+)=\E(\rho_A^-)$, where $\rho_{A}^{\pm}=\Pi_{\alpha}^{\pm}\rho_A \Pi_{\alpha}^{\pm}$ and $\sigma_{\alpha}=\Pi_{\alpha}^{+}-\Pi_{\alpha}^{-}$, then Bob's measurement outcomes will be uncorrelated with Alice's measurement outcomes, thus satisfying locality in time. This is not the case however for the examples considered in this section, since in the TPSM scenarios considered in the temporal Bell scenario Alice and Bob's measurement outcomes are always correlated, thus exhibiting ``non-local" effects in time.

\section{Discussion}

In this work, we have shown that a partial transpose of a bipartite density operator has the effect of transforming spatial correlations into temporal correlations, thus swapping the roles of space and time for bipartite quantum systems. Mathematically speaking, we established the Partial Transpose Theorem, which states that the partial transpositions $\rho_{AB}^{T_B}$ and $\rho_{AB}^{T_A}$ are both \emph{pseudo}-density operators, which provide a natural extension of the density operator formalism into the time domain. We illustrated our results in the explicit case of Bell scenarios, where we showed that the partial transpose maps correlations which violate spatial Bell inequalities to correlations which violate temporal Bell inequalities. 

As the partial transpose operation plays a fundamental role in the study of bipartite entanglement, the precise physical interpretation of partial transposition as established here should shed some light onto problems where traditional approaches have run out of steam. Moreover, the dynamical viewpoint of partial transposition as presented in this work suggests that the transpose map may play a role in contexts where gravitational effects become significant for quantum systems, such as inside the event horizon of a black hole. 

\emph{Acknowledgments}. JF would like to thank Ali Akil, Seok Hyung Lie, Xiangjing Liu, Arthur J. Parzygnat and Minjeong Song for useful discussions.

\bibliography{references}


\clearpage
\newpage

\title{Methods}
\author{testing}

\maketitle
\onecolumngrid
\vspace{1cm}

\begin{center}\large \textbf{Partial transpose as a spacetime swap} \\
\textbf{--- Supplemental Material ---}\\
\end{center}

\appendix

In this Supplemental Material, we prove the Partial Transpose Theorem, which appears as Theorem~\ref{MTX1} in the main text. As in the main text, $\A$ denotes the algebra of linear operators on a Hilbert space $\H_A$ associated with some quantum system $A$, and similarly for $\B$. We assume that orthonormal bases have been chosen for $\H_A$ and $\H_B$ so that we may identify $\A$ and $\B$ with matrix algebras $\M_d$ and $\M_n$, respectively. Density matrices will be referred to as \define{states}. The transpose map on a single system will be denoted by $T$, while partial transpose maps over system $A$ and $B$ will then be denoted by $T_A$ and $T_B$, respectively. The matrix units of $\A$ will be denoted by $E_{ij}$. Given a unitary operator $U$, the map $\rho \mapsto U\rho U^{\dag}$ will be denoted by $\Ad_U$. Given a linear map $\E:\A\to \B$, we recall that \define{Jamio\l kowski matrix} of $\E$ is the element $\J[\E]\in \A\otimes \B$ given by
\[
\J[\E]=\sum_{i,j}E_{ij}\otimes \E(E_{ji})\, ,
\]
and moreover, the map $\E\mapsto \J[\E]$ is a linear isomorphism. On the other hand, the \define{Choi matrix} of $\E$ is the operator $\mathscr{C}[\E]=\J[\E]^{T_A}$, so that
\[
\mathscr{C}[\E]=\sum_{i,j}E_{ij}\otimes \E(E_{ij})\, .
\]

\section{Some preliminary results from linear algebra}
In this appendix, we prove some elementary results from linear algebra to help facilitate the proof of the Partial Transpose Theorem.

\blem\label{lemma:block_matrix_zero_C}
If \( X = \begin{bmatrix} Y & Z \\ Z^\dagger & 0 \end{bmatrix} \geq 0 \) is a block matrix, then \( Z = 0 \).
\elem

\bprf
Assume \( Z \neq 0 \). Let \( u = (v, -tZv) \) where \( t \in \R \) and \( v \) is a suitable vector. We then have 
\[
\begin{aligned}
u^\dagger X u 
&= \begin{bmatrix} v^\dagger & -t (Zv)^\dagger \end{bmatrix} 
   \begin{bmatrix} Y & Z \\ Z^\dagger & 0 \end{bmatrix} 
   \begin{bmatrix} v \\ -tZv \end{bmatrix} = v^\dagger Y v - t v^\dagger Z Z^\dagger v - t (Zv)^\dagger Z^\dagger v + t^2 (Zv)^\dagger 0 (Zv) = v^\dagger Y v - 2t v^\dagger Z Z^\dagger v\ .
\end{aligned}
\]
Since \( X \geq 0 \), we analyze two cases:

\medskip

\noindent \textbf{Case 1.} If \( v^\dagger Y v = 0 \):  
Then \( u^\dagger X u = -2t v^\dagger Z Z^\dagger v \).  
Since \( Z \neq 0 \), there exists \( v \) such that \( v^\dagger Z Z^\dagger v > 0 \). Choosing \( t > 0 \)  makes \( u^\dagger X u < 0 \), contradicting \( X \geq 0 \).

\medskip

\noindent \textbf{Case 2.} If \( v^\dagger Y v > 0 \):  
Set \( t = \frac{v^\dagger Y v}{v^\dagger Z Z^\dagger v} \). Then  
\[
u^\dagger X u = v^\dagger Y v - 2 \cdot \frac{v^\dagger Y v}{v^\dagger Z Z^\dagger v} \cdot v^\dagger Z Z^\dagger v = -v^\dagger Y v < 0\ .
\]  
This also contradicts \( X \geq 0 \).
\\
Therefore, our assumption \( Z \neq 0 \) is false, i.e. \( Z = 0 \).
\eprf
\blem\label{lemma:block_matrix_positivity}
Let \( X = \sum_{i,j=1}^d E_{ij} \otimes X_{ij} \in \A \otimes \B \) be a block matrix. Then \( X \geq 0 \) if and only if  \( \sum_{i,j=1}^r E_{ij} \otimes X_{ij} \geq 0 \) for all \( r \in \{1,2,...d\} \).
\elem

\bprf

\medskip

\noindent \(\boldsymbol{(\Longrightarrow)}\)
Let \( X^{(r)} = \sum_{i,j=1}^r E_{ij} \otimes X_{ij} \) and define the projection matrix \( P_r = \begin{bmatrix} \mathds{1}_{rd} & \mathbf{0} \end{bmatrix} \in \M_{rd \times sn} \). Then
\[
X^{(r)} = P_r X P_r^\dagger
\]
For all non-zero vector \( x \in \C^{rd} \), we have
\[
\begin{aligned}
x^\dagger X^{(r)} x 
&= x^\dagger (P_r X P_r^\dagger) x = (P_r^\dagger x)^\dagger X (P_r^\dagger x) 
\end{aligned}
\]
Since \( X \geq 0 \), the right-hand side is non-negative. Thus, \( X^{(r)} \geq 0 \).

\medskip

\noindent \(\boldsymbol{(\Longleftarrow)}\) 
Let \( r = d \), then
\[
X^{(d)} = \sum_{i,j=1}^d E_{ij} \otimes X_{ij} = X
\]
By the condition for \( r = d \), \( X^{(d)} \geq 0 \), so \( X \geq 0 \).
\eprf
\blem\label{lemma:block_matrix_zero_blocks}
Let \( X = \sum_{i,j=1}^d E_{ij} \otimes X_{ij} \in \A \otimes \B \) be a block matrix with \( X \geq 0 \). If \( X_{ii} = 0 \), then \( X_{ik} = X_{ki} = 0 \) for all \( k \in \{1,2,...i\} \).
\elem

\bprf
 Assume $X_{ii} = 0$. By Lemma~\ref{lemma:block_matrix_positivity},  \( X^{(h)}=\sum_{r,s=1}^h E_{rs} \otimes X_{rs} \geq 0 \) for all \( h \in \{1,2,...d\} \) .Now,consider $X^{(h)}$ when \( h = i \), we have
\[
X^{(i)}=\sum_{r,s=1}^i E_{rs} \otimes X_{rs} 
= \begin{bmatrix} 
\sum_{r,s=1}^{i-1} E_{rs} \otimes X_{rs} & Y &...\\
Y^\dagger & E_{ii}\otimes X_{ii}&...\\
...&...&...
\end{bmatrix}=\begin{bmatrix} 
\sum_{r,s=1}^{i-1} E_{rs} \otimes X_{rs} & Y &...\\
Y^\dagger & 0&...\\
...&...&...
\end{bmatrix}\geq 0\, ,
\]
where (...) are all blocks of zero matrix and \( Y = \sum_{k=1}^{i-1}e_{k}\otimes X_{ki}\in\H_A\otimes\B\) .By Lemma \ref{lemma:block_matrix_zero_C}, a block matrix of the form \( \begin{bmatrix} * & Y \\ Y^\dagger & 0 \end{bmatrix} \geq 0 \) implies \( Y = 0 \).\\
Thus
\[
Y = Y^\dagger = 0 \implies X_{ki} = \overline{X_{ki}} = X_{ik} = 0\quad\text{for all}\ k\in\{1,2,...i\}\, ,
\]
as desired.
\eprf
\begin{lemma}\label{a}
For all $\rho\in\A$ and $X\in \A\otimes\B$, $T_A\big((\rho\otimes\mathds{1})X\big)=T_A\big(X\big)T_A\big((\rho\otimes\mathds{1})\big)$.

\end{lemma}

\begin{proof}
Let $X=\sum_{i,j=1}^dE_{ij}\otimes X_{ij}\in\A\otimes\B$ be a block representation of $X$. Then 
\begin{align*}
T_A\left((\rho\otimes\mathds{1})X\right)&=T_A\left((\rho\otimes\mathds{1})\left(\sum_{i,j=1}^dE_{ij}\otimes X_{ij}\right)\right)=T_A\left(\sum_{i,j=1}^d \rho E_{ij}\otimes X_{ij}\right)=\sum_{i,j=1}^d E_{ij}^T\rho^T\otimes X_{ij}\\
    &=\left(\sum_{i,j=1}^dE_{ij}^T\otimes X_{ij}\right)(\rho^T\otimes\mathds{1})=T_A\big(X\big)T_A\big((\rho\otimes\mathds{1})\big)\, ,
    \end{align*}
    as desired.
\end{proof}

\begin{lemma}\label{TTB}
    Let $X\in\A\otimes\B$, then $\Tr_B\circ T=T\circ\Tr_B$, $\Tr_B\circ T_A=T_A\circ\Tr_B$ and $\Tr_B\circ T_B=T_B\circ\Tr_B$.
\end{lemma}
\begin{proof}
    Consider the block decomposition $X=\sum_{i.j=1}^dE_{ij}\otimes X_{ij}$, we then have 
    \begin{align*}
        \Tr_B\circ T(X)&=\Tr_B\left[\left(\sum_{i.j=1}^dE_{ij}\otimes X_{ij}\right)^T\right]=\Tr_B\left[\sum_{i.j=1}^dE_{ji}\otimes X_{ij}^T\right]=\sum_{i,j=1}^d\Tr[X_{ij}^T]E_{ji}=\sum_{i,j=1}^d\Tr[X_{ij}]E_{ji}\\&=\left(\sum_{i,j=1}^d\Tr[X_{ij}]E_{ij}\right)^T=\left(\Tr_B\left[\sum_{i,j=1}^dE_{ij}\otimes X_{ij}\right]\right)^T=T\circ\Tr_B[X]\, ,
    \end{align*}
    as desired.Similarly, $\Tr_B\circ T_A=T_A\circ\Tr_B$ and $\Tr_B\circ T_B=T_B\circ\Tr_B$ are established \emph{mutatis mutandis}.
\end{proof}


\blem\label{lemma:partial-trace-unitary}
Let $U\otimes\mathds{1}\in\A\otimes\B$. Then $\Tr_B\circ\Ad_{(U\otimes\mathds{1})}=\Ad_U\circ\Tr_B$ .
\elem

\bprf
For all $X\in\A \otimes \B$, consider the block decomposition $X=\sum_{i,j=1}^d E_{ij} \otimes X_{ij}$ we then have 
\[
\begin{aligned}
\Tr_B\circ\Ad_{(U\otimes\mathds{1})}\big(X\big)&=\Tr_B\!\left[ (U^\dagger \otimes \mathds{1}) X (U \otimes \mathds{1}) \right] = \Tr_B\!\left[ (U^\dagger \otimes \mathds{1}) \left( \sum_{i,j=1}^d E_{ij} \otimes X_{ij} \right) (U \otimes \mathds{1}) \right] \\&= \Tr_B\!\left[ \sum_{i,j=1}^d U^\dagger E_{ij} U \otimes X_{ij} \right] = \sum_{i,j=1}^d \Tr[X_{ij}] U^\dagger E_{ij} U = U^\dagger \left( \sum_{i,j=1}^d \Tr[X_{ij}] E_{ij} \right) U \\
&= U^\dagger \Tr_B\!\left[\sum_{i,j=1}^d E_{ij} \otimes X_{ij} \right] U  = U^\dagger \Tr_B[X] U=\Ad_U\circ\Tr_B\big(X\big)\, ,
\end{aligned}
\]
as desired.
\eprf
\blem\label{lemma:partial_transpose_property}
Let \(\rho_{AB} \in \A \otimes \B\) be a state, let $\rho_A=\Tr_B[\rho_{AB}]\in\A$, let $U\in\A$ be a unitary such that \(U \rho_AU^\dagger = \mathrm{diag}(\lambda_1, \lambda_2, \ldots, \lambda_d)\), and let $\rho_{ij}$ be such that \(((U^\dagger)^T \otimes \mathds{1})\rho_{AB}^T(U^T \otimes \mathds{1}) = \sum_{i,j=1}^d E_{ij} \otimes \rho_{ij}\). If $\lambda_k = 0$ then $ \rho_{kk} = 0$ for all $k\in\{1,2,....d\}$ .
\elem

\bprf
Assume \(\lambda_k = 0\). Then 
\begin{align*}
0 &= \lambda_k = \left(\mathrm{diag}(\lambda_1, \lambda_2, \ldots, \lambda_d)\right)_{kk} =(U\rho_A U^\dagger)_{kk}= \left(U \Tr_B[\rho_{AB}]U^\dagger\right)_{kk} \stackrel{*}{=} \left(\Tr_B\left[(U \otimes \mathds{1})\rho_{AB}(U^\dagger \otimes \mathds{1})\right]\right)_{kk} \\
&=\left(\Tr_B\circ T\left(((U^\dagger)^T \otimes \mathds{1})\rho_{AB}^T(U^T \otimes \mathds{1})\right)\right)_{kk}\stackrel{**}{=} \left(T\circ\Tr_B\left(((U^\dagger)^T \otimes \mathds{1})\rho_{AB}^T(U^T \otimes \mathds{1})\right)\right)_{kk}\\
&=\left(T\circ\Tr_B\left(\sum_{i,j=1}^d E_{ij} \otimes \rho_{ij}\right)\right)_{kk}= \left(T\left(\sum_{i,j=1}^d \Tr[\rho_{ij}] E_{ij}\right)\right)_{kk} = \left(\left(\sum_{i,j=1}^d \Tr[\rho_{ij}] E_{ij}\right)^T\right)_{kk}\\
&= \left(\sum_{i,j=1}^d \Tr[\rho_{ij}] E_{ji}\right)_{kk}=\Tr[\rho_{kk}]\, ,
\end{align*}
where $*$ follows from Lemma~\ref{lemma:partial-trace-unitary} and $**$ follows from Lemma~\ref{TTB}. Since $U$ is a unitary and $\rho_{AB}$ is a state, we have
\[
\rho_{AB}\geq0\implies((U^\dagger)^T \otimes \mathds{1}) \rho_{AB}^T(U^T \otimes \mathds{1})^\dagger = \sum_{i,j=1}^d E_{ij} \otimes \rho_{ij}\implies\rho_{kk}\geq0\, .
\]
Now since \(\Tr[\rho_{kk}] = 0\) and $\rho_{kk}\geq 0$, we have
\[
0=\Tr[\rho_{kk}]=\sum_{i=1}^d (\rho_{kk})_{ii} \implies (\rho_{kk})_{ii} = 0 \,,\ i \in \{1,2,...d\} \implies \rho_{kk} = 0\, ,
\]
as desired.
\eprf
\blem\label{lemma:matrix_equation_solution}
Let \( \Lambda = \mathrm{diag}(\lambda_1, \lambda_2, \ldots, \lambda_d) \in \A \) be a diagonal matrix with $\lambda_i\geq 0$ for all $i\in \{1,2,...d\}$, and let $D = \sum_{i,j=1}^d E_{ij}\otimes D_{ij}\in \A\otimes \B$ be a positive semi-definite block matrix with $D_{ii} = 0$ for all $i$ such that $ \lambda_i= 0$. 
Then the following statements hold.
\begin{enumerate}[i.]
\item \label{LZXS1}
$Z$ is a solution of the matrix equation \( (\Lambda \otimes \mathds{1})Z + Z(\Lambda\otimes \mathds{1}) = D \) if and only if $Z$ is of the form
\be \label{ZXSX57}
Z = \sum_{\{i,j|\lambda_i + \lambda_j \neq 0\}} E_{ij} \otimes \frac{D_{ij}}{\lambda_i + \lambda_j} + \sum_{\{i,j|\lambda_i + \lambda_j = 0\}} E_{ij} \otimes W_{ij}\, .
\ee
\item \label{LZXS2}
The matrix $\sum_{\{i,j|\lambda_i+\lambda_j\neq 0\}} E_{ij} \otimes \frac{D_{ij}}{\lambda_i + \lambda_j}$ is positive semi-definite.
\item \label{LZXS3}
If $\lambda_i > 0$ for all $i\in \{1,\ldots,d\}$, then the matrix equation $(\Lambda \otimes \mathds{1})Z + Z(\Lambda\otimes \mathds{1}) = D$ has a unique solution given by $Z=\sum_{i,j=1}^{d} E_{ij} \otimes \frac{D_{ij}}{\lambda_i + \lambda_j}$, which is positive semi-definite. 
\end{enumerate}
\elem

\bprf
First, if $\lambda_i+\lambda_j=0$ then $\lambda_i=\lambda_j=0$ since $\lambda_i\geq0$ for all $i\in \{1,\ldots,d\}$. And since $\lambda_i=\lambda_j=0$ implies $D_{ii}=D_{jj}=0$ by assumption, it follows from Lemma~\ref{lemma:block_matrix_zero_blocks} that $D_{ij}=0$ whenever $\lambda_i+\lambda_j=0$.\\

\underline{Item~\ref{LZXS1}}: $(\implies)$ Suppose $Z$ is a solution of the matrix equation $(\Lambda \otimes \mathds{1})Z + Z(\Lambda\otimes \mathds{1}) = D$, and let  \( Z = \sum_{i,j=1}^d E_{ij} \otimes Z_{ij} \) be a block representation of $Z$. We then have
\[
\begin{aligned}
&(\Lambda \otimes \mathds{1}) \left( \sum_{i,j=1}^d E_{ij} \otimes Z_{ij} \right) + \left( \sum_{i,j=1}^d E_{ij} \otimes Z_{ij} \right) (\Lambda \otimes \mathds{1}) = \sum_{i,j=1}^d E_{ij} \otimes D_{ij} \\
\implies &\sum_{i,j=1}^d \lambda_i E_{ij} \otimes Z_{ij} + \sum_{i,j=1}^d \lambda_j E_{ij} \otimes Z_{ij} = \sum_{i,j=1}^d E_{ij} \otimes D_{ij} \\
\implies &\sum_{i,j=1}^d (\lambda_i + \lambda_j) E_{ij} \otimes Z_{ij} = \sum_{i,j=1}^d E_{ij} \otimes D_{ij} \\
\implies &(\lambda_i + \lambda_j) Z_{ij} = D_{ij}
\end{aligned}
\]
If \( \lambda_i + \lambda_j \neq 0 \), then $Z_{ij} = \frac{D_{ij}}{\lambda_i + \lambda_j}$. And if $\lambda_i + \lambda_j = 0$, then $D_{ij}=0$,  thus there are no constraints on $Z_{ij}$. Therefore, $Z$ is necessarily of the form
\[
Z = \sum_{\{i,j|\lambda_i + \lambda_j \neq 0\}} E_{ij} \otimes \frac{D_{ij}}{\lambda_i + \lambda_j} + \sum_{\{i,j|\lambda_i + \lambda_j = 0\}} E_{ij} \otimes W_{ij}\, ,
\]
as desired.\\
$(\impliedby)$ 
Suppose $Z$ is of the form \eqref{ZXSX57}. Since $D_{ij}=0$ whenever $\lambda_i+\lambda_j=0$, it follows that $D=\sum_{i,j=1}^dE_{ij}\otimes D_{ij}=\sum_{\{i,j\,|\,\lambda_i + \lambda_j \neq 0\}} E_{ij} \otimes D_{ij}$. We then have 
\begin{align*}
(\Lambda \otimes \mathds{1})Z
&=\left(\sum_{k=1}^d E_{kk} \otimes \lambda_k \mathds{1}\right)\left(\sum_{\{i,j,|,\lambda_i + \lambda_j \neq 0\}} E_{ij} \otimes \frac{D_{ij}}{\lambda_i + \lambda_j} + \sum_{\{i,j,|,\lambda_i + \lambda_j = 0\}} E_{ij} \otimes W_{ij}\right)\\
&=\sum_{\{i,j,|,\lambda_i + \lambda_j \neq 0\}}\sum_{k=1}^d E_{kk}E_{ij} \otimes \frac{\lambda_k D_{ij}}{\lambda_i + \lambda_j} + \sum_{\{i,j,|,\lambda_i + \lambda_j = 0\}}\sum_{k=1}^d E_{kk}E_{ij} \otimes \lambda_k W_{ij} \\
&=\sum_{\{i,j,|,\lambda_i + \lambda_j \neq 0\}}\sum_{k=1}^d \delta_{ki}E_{ij} \otimes \frac{\lambda_k D_{ij}}{\lambda_i + \lambda_j} +\sum_{\{i,j,|,\lambda_i + \lambda_j = 0\}}\sum_{k=1}^d \delta_{ki}E_{ij} \otimes \lambda_k W_{ij} \\
&=\sum_{\{i,j,|,\lambda_i + \lambda_j \neq 0\}} E_{ij} \otimes \frac{\lambda_i D_{ij}}{\lambda_i + \lambda_j} + \sum_{\{i,j,|,\lambda_i + \lambda_j = 0\}} E_{ij} \otimes \lambda_i W_{ij}\, .
\end{align*}
Similarly, we have 
\[
Z(\Lambda \otimes \mathds{1})=\sum_{\{i,j\,|\,\lambda_i + \lambda_j \neq 0\}} E_{ij} \otimes \frac{\lambda_j D_{ij}}{\lambda_i + \lambda_j} + \sum_{\{i,j\,|\,\lambda_i + \lambda_j = 0\}} E_{ij} \otimes \lambda_j W_{ij}\, ,
\]
thus
\begin{align*}
    (\Lambda \otimes \mathds{1})Z+Z(\Lambda \otimes \mathds{1})&=\sum_{\{i,j\,|\,\lambda_i + \lambda_j \neq 0\}} E_{ij} \otimes \frac{(\lambda_i+\lambda_j) D_{ij}}{\lambda_i + \lambda_j} + \sum_{\{i,j\,|\,\lambda_i + \lambda_j = 0\}} E_{ij} \otimes (\lambda_i+\lambda_j) W_{ij}\\
    &=\sum_{\{i,j\,|\,\lambda_i + \lambda_j \neq 0\}} E_{ij} \otimes D_{ij}=D\, ,
\end{align*}
as desired.

\underline{Item~\ref{LZXS2}}: Let \( v = \sum_{k=1}^d e_k \otimes v_k \in \H_A \otimes \H_B \) be a non-zero block vector and let $X=\sum_{\{i,j\,|\,\lambda_i + \lambda_j \neq 0\}} E_{ij} \otimes \frac{D_{ij}}{\lambda_i + \lambda_j}$. Since $D_{ij}=0$ whenever $\lambda_i+\lambda_j=0$, $D=\sum_{i,j=1}^dE_{ij}\otimes D_{ij}=\sum_{\{i,j\,|\,\lambda_i + \lambda_j \neq 0\}} E_{ij} \otimes D_{ij}$\ . So we then have

\begin{align*}
v^\dagger X v
&= \left( \sum_{k=1}^d e_k^\dagger \otimes v_k^\dagger \right)
\left( \sum_{\{i,j,|,\lambda_i + \lambda_j \neq 0\}} E_{ij} \otimes \frac{D_{ij}}{\lambda_i + \lambda_j} \right)
\left( \sum_{l=1}^d e_l \otimes v_l \right)= \sum_{\{i,j,|,\lambda_i + \lambda_j \neq 0\}}\sum_{k,l=1}^d \left( e_k^\dagger E_{ij} e_l \right)
\left( v_k^\dagger \frac{D_{ij}}{\lambda_i + \lambda_j} v_l \right) \\
&= \sum_{\{i,j,|,\lambda_i + \lambda_j \neq 0\}} \frac{1}{\lambda_i + \lambda_j} v_i^\dagger D_{ij} v_j= \sum_{\{i,j,|,\lambda_i + \lambda_j \neq 0\}} v_i^\dagger D_{ij} v_j \int_0^\infty e^{-t(\lambda_i + \lambda_j)} dt \\
&= \int_0^\infty \sum_{\{i,j,|,\lambda_i + \lambda_j \neq 0\}} v_i^\dagger e^{-t\lambda_i} D_{ij} e^{-t\lambda_j} v_j dt = \int_0^\infty \sum_{\{i,j,|,\lambda_i + \lambda_j \neq 0\}} v_i(t)^\dagger D_{ij} v_j(t), dt \\
&= \int_0^\infty \sum_{\{i,j,|,\lambda_i + \lambda_j \neq 0\}}\sum_{k,l=1}^d \left( e_k^\dagger E_{ij} e_l \right)
\left( v_k(t)^\dagger D_{ij} v_l(t) \right) dt \\
&= \int_0^\infty \left( \sum_{k=1}^d e_k^\dagger \otimes v_k(t)^\dagger \right)
\left( \sum_{\{i,j,|,\lambda_i + \lambda_j \neq 0\}} E_{ij} \otimes D_{ij} \right)
\left( \sum_{l=1}^d e_l \otimes v_l(t) \right) dt = \int_0^\infty v(t)^\dagger D v(t) dt\geq0 \, ,
\end{align*}
where $v_i(t)=v_i e^{-t\lambda_i}$ .

\underline{Item~\ref{LZXS3}}: If $\lambda_i>0$ for all $i\in \{1,2,...d\} $ then it follows immediately from item~\ref{LZXS1} that $Z=\sum_{i,j=1}^{d} E_{ij} \otimes \frac{D_{ij}}{\lambda_i + \lambda_j}$ is the unique solution of the matrix equation $(\Lambda\otimes \mathds{1})Z+Z(\Lambda\otimes \mathds{1})=D$. The fact that $Z\geq 0$ then follows from item~\ref{LZXS2}.

\eprf

\blem\label{lemma:trace_preserving_map}
Let \( \map{E} : \A \to \B \) be a linear map. Then \( \map{E} \) is a trace-preserving map if and only if \( \Tr_B\big[\mathscr{C}[\map{E}]\big] = \mathds{1} \).
\elem

\bprf

\medskip

\noindent \(\boldsymbol{(\Longrightarrow)}\) Assume $\E$ is trace-preserving. We then have
\[
\Tr_B\big[\mathscr{C}[\map{E}]\big]= \Tr_B\left[ \sum_{i,j=1}^d E_{ij} \otimes \map{E}(E_{ij}) \right]= \sum_{i,j=1}^d \Tr\big[\map{E}(E_{ij})\big] E_{ij}= \sum_{i,j=1}^d \Tr[E_{ij}] E_{ij}= \sum_{i,j=1}^d \delta_{ij} E_{ij}= \sum_{i=1}^d E_{ii}= \mathds{1}\, ,
\]
as desired.

\noindent \(\boldsymbol{(\Longleftarrow)}\) Assume \( \Tr_B\big[\mathscr{C}[\map{E}]\big] = \mathds{1} \). We then have
\[
\sum_{i,j=1}^d \Tr\big[\map{E}(E_{ij})\big] E_{ij} = \mathds{1} \implies \Tr\big[\map{E}(E_{ij})\big] = \delta_{ij}\ .
\]
Now let \( X = \sum_{i,j=1}^d x_{ij} E_{ij} \in \A \) be arbitrary. We then have
\[
\Tr\big[\map{E}(X)\big]= \Tr\left[ \map{E}\left( \sum_{i,j=1}^d x_{ij} E_{ij} \right) \right]= \sum_{i,j=1}^d x_{ij} \Tr\big[\map{E}(E_{ij})\big]= \sum_{i,j=1}^d x_{ij} \delta_{ji}= \sum_{i=1}^d x_{ii} = \Tr[X]\, ,
\]
thus \( \map{E} \) is trace-preserving, as desired.
\eprf
\blem\label{lemma:block_matrix_trace}
Let \( X = \sum_{i,j=1}^d E_{ij} \otimes X_{ij} \in \A \otimes \B \) be a block matrix. Then $\Tr[X_{ij}] = \left(\Tr_B[X]\right)_{ij}$.

\elem

\bprf
Indeed,
\[
\begin{aligned}
\left(\Tr_B[X]\right)_{ij} 
&= \left( \Tr_B \left[ \sum_{i,j=1}^d E_{ij} \otimes X_{ij} \right] \right)_{ij} = \left( \sum_{i,j=1}^d \Tr[X_{ij}] E_{ij} \right)_{ij}= \Tr[X_{ij}]\, ,
\end{aligned}
\]
as desired.
\eprf

\section{Proof of the Partial Transpose Theorem}
\bt[The Partial Transpose Theorem]
Let $\rho_{AB}\in \A\otimes \B$ be a bipartite density operator. Then there exists a quantum channel $\E:\A\to \B$ such that 
\be \label{PTT2X}
\rho_{AB}^{T_B}=\E\star \rho_A \quad \text{and} \quad \rho_{AB}^{T_A}=\overline{\E}\star \rho_A^T\, .
\ee
where $\overline{\E}=T\circ \E\circ T$. In particular, if $\rho_A=\Tr_B[\rho_{AB}]$ is of full rank, then the channel $\E$ is unique.
\et

\bprf
We first show that there exists a linear map $\E:\A\to \B$ such that 
\be \label{MNXEQCX57}
\rho_{AB}^{T_B}=\E\star \rho_A\, ,
\ee
and then we will show that there exists such an $\E$ that is completely positive (CP) and trace-preserving (TP). For this, first note that since
\[
\E\star \rho_A= \frac{1}{2}\Big\{\rho_A\otimes \mathds{1},\J[\E]\Big\}\, ,
\]
it follows that there exists a linear map $\E$ satisfying Eq.~(\ref{MNXEQCX57}) if and only if $X=\J[\E]$ is a solution to the matrix equation
\be \label{SYLXVSTRX67}
(\rho_A\otimes \mathds{1})X+X(\rho_A\otimes \mathds{1})=2T_B(\rho_{AB})\, .
\ee
Now note that $X$ is a solution of Eq.~(\ref{SYLXVSTRX67}) if and only if $Z$ is a solution of the equation 
\be \label{AUXEQX87}
(\Lambda\otimes \mathds{1})Z+Z(\Lambda\otimes \mathds{1})=D\, ,
\ee
where $\Lambda$ is a diagonalization of $\rho_A$, so that $\Lambda=U\rho_AU^{\dag}$ for some unitary $U\in \A$, $Z=\Phi(X)$, $\Lambda\otimes\mathds{1}=\Phi(\rho_A\otimes\mathds{1})$ and $D=\Phi(2T_B(\rho_A))$, where $\Phi=T_A\circ\Ad_{(U\otimes\mathds{1})}$. By Lemma~\ref{a}, it follows that
\[
D=((U^\dagger)^T\otimes\mathds{1})2\rho_{AB}^T(U^T\otimes\mathds{1})\, ,
\]
and by Lemma~\ref{lemma:partial_transpose_property}, it follows that $D$ satisfies the conditions of Lemma~\ref{lemma:matrix_equation_solution}, which ensures that there indeed exists a solution $Z$ of Eq.~(\ref{AUXEQX87}). Moreover, for every such solution $Z$ it follows that $\E=\J^{-1}[\Phi^{-1}(Z)]$ is a solution of Eq.~(\ref{MNXEQCX57}). In the case that $\rho_A$ is of full rank, it follows from Item~\ref{LZXS3} of Lemma~\ref{lemma:matrix_equation_solution} that Eq.~(\ref{AUXEQX87}) admits a unique solution $Z$, from which it follows that Eq.~(\ref{MNXEQCX57}) admits a unique solution $\E=\J^{-1}[\Phi^{-1}(Z)]$. 

We now show that we can always find a solution $Z$ to Eq.~(\ref{AUXEQX87}) such that $\E=\J^{-1}[\Phi^{-1}(Z)]$ is completely positive (CP) and trace-preserving (TP).
\\

\noindent \textbf{Completely Positivity.} We will now show there exists a solution $Z$ of \eqref{AUXEQX87} such that $\E=\mathscr{J}^{-1}[\Phi^{-1}(Z)]$ is completely positive. Since a linear map $\E$ is completely positive if and only if its Choi matrix $\mathscr{C}[\E]$ is positive semi-definite, we will now show there exists a solution $Z$ such that $\mathscr{C}[\E]\geq0$, where $\E=\mathscr{J}^{-1}[\Phi^{-1}(Z)]$. For this, first note that $\E=\mathscr{J}^{-1}[\Phi^{-1}(Z)]$ implies $Z=((U^\dagger)^T \otimes \mathds{1})\Choi[\map{E}](U^T \otimes \mathds{1})$. Now according to Item~\ref{LZXS1} of Lemma~\ref{lemma:matrix_equation_solution}, $Z$ is a solution of Eq.~(\ref{AUXEQX87}) if and only if it is of the form
\be\label{CP1}
Z
= \sum_{\{i,j|\lambda_i + \lambda_j \neq 0\}} E_{ij} \otimes \frac{D_{ij}}{\lambda_i + \lambda_j} + \sum_{\{i,j|\lambda_i + \lambda_j = 0\}} E_{ij} \otimes W_{ij}\, .
\ee
By Item~\ref{LZXS2} of Lemma~\ref{lemma:matrix_equation_solution}, the matrix $\sum_{\{i,j|\lambda_i + \lambda_j \neq 0\}} E_{ij} \otimes \frac{D_{ij}}{\lambda_i + \lambda_j}\geq0$, thus by imposing $W_{ii}\geq0$ and $W_{ij}=0$ for all $i\neq j$, we then have 
\[\sum_{\{i,j|\lambda_i+\lambda_j=0\}}E_{ij}\otimes W_{ij}=\sum_{\{i|\lambda_i=0\}}E_{ii}\otimes W_{ii}\geq0\, ,\]
thus
\be \label{CPXS1}
Z=(U^T \otimes \mathds{1})\Choi[\map{E}](U^T \otimes \mathds{1})^\dagger 
= \underbrace{\sum_{\{i,j|\lambda_i + \lambda_j \neq 0\}} E_{ij} \otimes \frac{D_{ij}}{\lambda_i + \lambda_j}}_{\geq 0} 
+ \underbrace{\sum_{\{i,j|\lambda_i + \lambda_j = 0\}} E_{ij} \otimes W_{ij}}_{\geq 0} \geq 0\, .
\ee 
Since $U$ is unitary, it then follows that $\Choi[\map{E}] \geq 0$, as desired.
\\

\noindent \textbf{Trace Preserving.}  By Lemma~\ref{lemma:trace_preserving_map} we have that $\E$ is trace-preserving if and only if $\Tr_B\big[\mathscr{C}[\E]\big]=\mathds{1}$, thus we now prove $\Tr_B\big[\mathscr{C}[\E]\big]=\mathds{1}$. By Eq.~(\ref{CPXS1}), we have 
\be\label{TP1}
((U^\dagger)^T \otimes \mathds{1})\Choi[\map{E}](U^T \otimes \mathds{1}) 
= \sum_{\{i,j|\lambda_i + \lambda_j \neq 0\}} E_{ij} \otimes \frac{D_{ij}}{\lambda_i + \lambda_j} + \sum_{\{i,j|\lambda_i + \lambda_j = 0\}} E_{ij} \otimes W_{ij}\, ,
\ee
where we recall that we have imposed the conditions $W_{ii}=0$ and $W_{ij}=0$ for all $i\neq j$. We now further impose the condition $\Tr_B[W_{ii}]=1$. Now by taking the partial trace $\Tr_B$ of both sides of Eq.~(\ref{TP1}) and applying Lemma~\ref{lemma:partial-trace-unitary} we then have
\begin{align*}
    (U^\dagger)^T\mathscr{C}[\E]U^T&=\sum_{\{i,j|\lambda_i+\lambda_j\neq0\}}\Tr\left[\frac{D_{ij}}{\lambda_i+\lambda_j}\right]E_{ij}+\sum_{\{i|\lambda_i0\}}\Tr\left[W_{ii}\right]E_{ii}\\
    &=\sum_{\{i,j|\lambda_i+\lambda_j\neq0\}}\frac{1}{\lambda_i+\lambda_j}\Tr\left[D_{ij}\right]E_{ij}+\sum_{\{i|\lambda_i0\}}E_{ii}\\
    &=\sum_{\{i,j|\lambda_i+\lambda_j\neq0\}}\frac{1}{\lambda_i+\lambda_j}\left(\Tr_B[D] \right)_{ij}E_{ij}+\sum_{\{i|\lambda_i0\}}E_{ii}\tag{Lemma~\ref{lemma:block_matrix_trace}}\\
    &=\sum_{\{i,j|\lambda_i+\lambda_j\neq0\}}\frac{1}{\lambda_i+\lambda_j}\left(\Tr_B\left[((U^\dagger)^T\otimes\mathds{1})2\rho_{AB}^T(U^T\otimes \mathds{1})\right] \right)_{ij}E_{ij}+\sum_{\{i|\lambda_i0\}}E_{ii}\\
    &=\sum_{\{i,j|\lambda_i+\lambda_j\neq0\}}\frac{2}{\lambda_i+\lambda_j}\left(\Tr_B\circ T\left((U\otimes\mathds{1})\rho_{AB}(U^\dagger\otimes \mathds{1})\right) \right)_{ij}E_{ij}+\sum_{\{i|\lambda_i0\}}E_{ii}\\
    &=\sum_{\{i,j|\lambda_i+\lambda_j\neq0\}}\frac{2}{\lambda_i+\lambda_j}\left(T\circ\Tr_B \left((U^\dagger\otimes\mathds{1})\rho_{AB}(U^\dagger\otimes \mathds{1})\right) \right)_{ij}E_{ij}+\sum_{\{i|\lambda_i0\}}E_{ii}\tag{Lemma~\ref{TTB}}\\
    &=\sum_{\{i,j|\lambda_i+\lambda_j\neq0\}}\frac{2}{\lambda_i+\lambda_j}\left(T \left(U\Tr_B[\rho_{AB}]U^\dagger\right) \right)_{ij}E_{ij}+\sum_{\{i|\lambda_i0\}}E_{ii}\tag{Lemma~\ref{lemma:partial-trace-unitary}}\\
    &=\sum_{\{i,j|\lambda_i+\lambda_j\neq0\}}\frac{2}{\lambda_i+\lambda_j}\left(T(U \rho_A U^\dagger) \right)_{ij}E_{ij}+\sum_{\{i|\lambda_i0\}}E_{ii}\\
    &=\sum_{\{i,j|\lambda_i+\lambda_j\neq0\}}\frac{2}{\lambda_i+\lambda_j}\left(\Lambda^T \right)_{ij}E_{ij}+\sum_{\{i|\lambda_i0\}}E_{ii}\\
    &=\sum_{\{i,j|\lambda_i+\lambda_j\neq0\}}\frac{2}{\lambda_i+\lambda_j}\left(\Lambda \right)_{ij}E_{ij}+\sum_{\{i|\lambda_i0\}}E_{ii}\\
    &=\sum_{\{i|\lambda_i\neq0\}}\frac{2\lambda_i}{2\lambda_i}E_{ii}+\sum_{\{i|\lambda_i0\}}E_{ii}\\
    &=\sum_{i=1}^dE_{ii}\\
    &=\mathds{1}\, .
\end{align*}
Since $U$ is unitary, it then follows that $\Tr_B\big[\mathscr{C}[\E]\big]=\mathds{1}$, as desired.

Now if $\rho_A=\Tr_B(\rho_{AB})$ is of full rank, then by Item~\ref{LZXS3} of Lemma~\ref{lemma:matrix_equation_solution} we have that the solution Z of Eq.~(\ref{AUXEQX87}) is unique, positive semi-definite and takes the form
\[Z=\sum_{i,j=1}^dE_{ij}\otimes\frac{D_{ij}}{\lambda_i+\lambda_j}\geq0\, .\]
Since $Z=((U^\dagger)^T\otimes\mathds{1})\mathscr{C}[\E](U^T\otimes\mathds{1})$, we have 
\be\label{UNI}
((U^\dagger)^T\otimes\mathds{1})\mathscr{C}[\E](U^T\otimes\mathds{1})=\sum_{i,j=1}^dE_{ij}\otimes\frac{D_{ij}}{\lambda_i+\lambda_j}\, .
\ee
Since $U$ is unitary, it follows that $\mathscr{C}[\E]\geq0$, i.e., $\E$ is also completely positive. Similarly, taking the partial trace $\Tr_B$ of both sides of Eq.~(\ref{UNI}) and then simplifying yields 
\[
(U^\dagger)^T\Tr_B\big[\mathscr{C}[\E]\big]U^T=\mathds{1}\, ,
\]
thus $\Tr_B\big[\mathscr{C}[\E]\big]=\mathds{1}$, i.e., $\E$ is also trace-preserving.
\\
\\
Finally, since we know that there exists a CPTP map $\E:\A\to \B$ such that 
\[
\rho_{AB}^{T_B}=\E\star \rho_A\, ,
\]
it follows that
\begin{align*}
\rho_{AB}^{T_A}&=(\E\star \rho_A)^T=\sum_{i,j}^d (E_{ij}\rho_A+\rho_AE_{ij})^T\otimes \E(E_{ji})^T \\
&=\sum_{i,j}^d(E_{ji}\rho_A^T+\rho_A^TE_{ji})\otimes (T\circ \E\circ T)(E_{ij}) \\
&=(T\circ \E\circ T)\star \rho_A^T \\
&=\overline{\E}\star \rho_A^T\, , 
\end{align*}
where $\overline{\E}=T\circ \E\circ T$. Moreover, since 
\[
\mathscr{C}[\overline{\E}]=\mathscr{C}[T\circ \E\circ T]=\sum_{i,j}E_{ij}\otimes (T\circ \E\circ T)(E_{ij})=\sum_{i,j}E_{ij}\otimes \E(E_{ji})^{T}=\sum_{i,j}E_{ji}^{T}\otimes \E(E_{ji})^{T}=\mathscr{C}[\E]^T\geq 0\, ,
\]
it follows that $\overline{\E}$ is completely positive. Moreover, since $\E$ is trace-preserving and the transpose is trace-preserving, it follows that $\overline{\E}$ is also trace-preserving, thus $\overline{\E}$ is CPTP. This completes the proof.
\eprf

\end{document}